\documentclass[11pt,a4paper]{article}

\usepackage[a4paper,top=1.2in, bottom=1.2in, left=1in, right=1in]{geometry}
\usepackage{commath}
\usepackage{amsmath,paralist,amsthm,mathpazo,pspicture}
\usepackage{amsfonts,amssymb}
\usepackage{lmodern}
\usepackage{ulem}
\usepackage{setspace}
\usepackage{enumerate}
\usepackage{pgfplots}
\usepackage{float}
\usepackage{enumitem}
\usepackage{tikz}
\usepackage{multirow,array}
\usepackage{pbox}
\usepackage{graphicx}
\usepackage[utf8]{inputenc}
\usepackage{amssymb}
\usepackage{adjustbox}
\usepackage{fancyref}
\usepackage[utf8]{inputenc}
\usepackage{scalerel}
\usepackage[round]{natbib}
\usepackage{lineno}
\usepackage{setspace}
\usepackage{subfigure}
\usepackage{tikz}
\usetikzlibrary{positioning}
\usepackage{float}
\usepackage{caption}
\usepackage{tabstackengine}
\usepackage{graphics,graphicx}
\usepackage{pstricks,pst-node,pst-tree,pst-pdf,pstricks-add}
\usepackage{multicol}
\usepackage{parallel}
\allowdisplaybreaks

\graphicspath{ {E:/images/} }

\usetikzlibrary{positioning}

\newtheorem{theorem}{Theorem}
\newtheorem{proposition}{Proposition}

\newtheorem{lemma}{Lemma}

\newtheoremstyle{as}
{\topsep}
{\topsep}
{\itshape}
{}
{\scshape}
{.}
{ }
{\thmname{\textbf{\textit{#1}}}\thmnumber{ $\mathbf{(#2)}$}\thmnote{ (#3)}}       
\theoremstyle{as}

\newtheoremstyle{pro}
{\topsep}
{\topsep}
{\itshape}
{}
{\scshape}
{.}
{ }
{\thmname{\textbf{\textit{#1}}}\thmnumber{ $\mathbf{#2}$}\thmnote{ (#3)}}%
\theoremstyle{pro}

\newtheoremstyle{lem}
{\topsep}
{\topsep}
{\itshape}
{}
{\scshape}
{.}
{ }
{\thmname{\textbf{\textit{#1}}}\thmnumber{ $\mathbf{#2}$}\thmnote{ (#3)}}%
\theoremstyle{lem}

\usepackage{subfigure}
\usepackage{tikz}

\begin{document}
\centerline{\textbf{\Large The Uniformed Patroller Game}}
\bigskip
\centerline{\textbf{Steve Alpern}\footnote{\scriptsize ORMS Group, Warwick Business School, University of Warwick, Coventry, CV4 7AL, UK, steve.alpern@wbs.ac.uk}, \textbf{Paul Chleboun}\footnote{\scriptsize Department of Statistics, University of Warwick, Coventry, CV4 7AL, UK, Paul.I.Chleboun@warwick.ac.uk}, \textbf{Stamatios Katsikas}\footnote{\scriptsize Centre for Complexity Science, University of Warwick, Coventry, CV4 7AL, UK, stamkatsikas@gmail.com}\footnote{\scriptsize International Education Institute, University of St Andrews, St Andrews, KY16 9DJ, UK, sk296@st-andrews.ac.uk}}
\smallskip 
\centerline{\today}
\bigskip
\begin{abstract}
\noindent  
Patrolling Games were introduced by \cite{Alp-1} to model the adversarial problem where a mobile Patroller can thwart an attack at some location only by visiting it during the attack period, which has a prescribed integer duration $m$. In this note we modify the problem by allowing the Attacker to go to his planned attack location early and observe the presence or the absence there of the Patroller (who wears a uniform). To avoid being too predictable, the Patroller may sometimes remain at her base when she could have been visiting a possible attack location. The Attacker can then choose to delay attacking for some number of periods $d$ after the Patroller leaves his planned attack location. As shown here, this extra information for the Attacker can reduce thwarted attacks by as much as a factor of four in specific models. Our main finding, conjectured by the Associate Editor, is that the attack should begin in the second period the Patroller is away ($d=2$) and that the Patroller should never attack the same location in consecutive periods.
\end{abstract}

\section{Introduction}\label{Section1}
\indent

Patrolling games were introduced by \cite{Alp-1} to model the adversarial problem where a mobile Patroller can  thwart an attack at some location only by visiting (inspecting) there during the attack period. The attack period has a prescribed duration that is an integer number $m$ of consecutive periods (the term attack duration will be used  interchangeably with the term attack difficulty). In this paper we introduce a modification of the original model, whereby the Attacker can arrive at his intended attack location early and observe the presence or absence there of the Patroller. We say that the Patroller is wearing a uniform in order to stress that she is observable, but only at short range, when being at the location of the Attacker. This property gives the Attacker an additional choice variable, that is, the number of periods $d$ to delay his attack after the Patroller leaves the attack location (the term attack delay will be used  interchangeably with the term waiting time). If the Patroller comes back before these $d$ periods, the Attacker must restart his count. Thus the attack begins in the first period in which the Patroller has been recorded to be away for $d$ consecutive periods.

To illustrate the parameter $d$, suppose that the Patroller's presence at the attack location is indicated by a $1$, and her absence is indicated by a $0$. Suppose also that the Attacker waits until the Patroller arrives at the attack location, and after that her presence-absence sequence that he records is given by $11010100\dots$ (see Table \ref{Table1}). 

\setlength{\tabcolsep}{13pt}
\renewcommand{\arraystretch}{1}
\begin{table}[H]
\centering
\begin{tabular}{ c  c  c  c  c  c  c  c }
\hline\hline
$1$ & $11$ & $110$ & $1101$ & $11010$ & $110101$ & $1101010$ & $11010100$ \\ \hline
\textit{wait} & \textit{wait} & \textit{wait} & \textit{wait} & \textit{wait} & \textit{wait} & \textit{wait} & \textit{attack!} \\ 
\hline\hline
\end{tabular}
\caption{The Attacker attacks after $d=2$ periods of the Patroller's absence}
\label{Table1}
\end{table}

If for example the Attacker's waiting time is $d=2$, then his decision as to whether or not to attack in each subsequent period is illustrated in Table \ref{Table1}. If also the attack duration is $m=3$, then the attack will be successful only if the Patroller's sequence continues with two more $0$'s, that is, $1101010000\dots$. Note that the attack can begin in the same period the Patroller's absence has been observed (if $d=1$). Thus, we take $m\geq 2$ as otherwise the Attacker could simply win the game by attacking as soon as the Patroller is not present. 

In our model the Patroller has a base (not considered a `location') from which she can reach any of the $n$ locations she has to defend in a single period. 
In principle, the Patroller could visit a location in consecutive periods and her base could also be attacked, but we show that such a behavior is not optimal. 
Thus after visiting a location, the Patroller returns to her base. 
One might think that the Patroller should always visit some location immediately after returning to her base, but in fact she might want to stay at her base for consecutive periods to produce additional uncertainty for the Attacker.
We determine the optimal probability $r$ in which to remain at her base for an additional period, as a function of the number of locations $n$ and the attack duration $m$ and delay $d$. Our main general finding (Theorem \ref{Theorem1}) is that $d=2$ is always an optimal value of the delay in attacking, and that the Patroller should never inspect (stay at) a location in consecutive periods.
The later result is proved using what is known as a coupling argument.

We assume that the Patroller knows her current location, but she does not remember her previous sequence of locations (i.e. she is Markovian). The assumption of a Markovian Patroller is a standard approach in the literature of patrolling games, particularly when Patrollers are robots, see, e.g., \cite{Basilico-1,Basilico-2,Collins}. For the Attacker, we determine his optimal delay $d$. We call this game the \textit{The Uniformed Patroller Game}.

After solving the game, it is straightforward to compare its value (i.e. the optimal probability that an attack is thwarted) to the corresponding easily tractable version where the Patroller wears no uniform. We find, for example, that for short attack duration $m=2$ and for a large number of locations $n$, wearing a uniform reduces the interception probability for an attack by a factor of $4$. Hence, we believe that the common occurrence of uniformed Patrollers is likely explained not by their ability to intercept attacks, but to the effect on uniforms revealing the presence of a Patroller and thereby deterring attacks from taking place at all. In our model, as in previous patrolling models, the Attacker will never be deterred from carrying out an attack. Deterrence could be modeled in future work. When there are multiple Patrollers with prescribed paths (e.g. a single direction around a circle) but varying start times, it has been found in a thought provoking paper by \cite{Lin-3} that, unlike our result, wearing a uniform does not hurt the Patroller(s). In this sense our paper and that by \cite{Lin-3} show that the effect of wearing a uniform depends greatly on whether fixed or unpredictable patrolling paths are used. This is where the Markovian nature of the paths in our model is particularly important.

Patrolling problems have been studied for a long time, see e.g. \cite{Morse}, but only from the Patroller's viewpoint. A game theoretic approach, modelling an adversarial Attacker who wants to infiltrate or attack a network at a node of his choice, has only recently been introduced by \cite{Alp-1}. The techniques developed there were also applied to the class of line networks by \cite{Pap-1}, with the interpretation of patrolling a border. A continuous version of patrolling games on a network has been given by \cite{Garrec}. Research of a similar reasoning includes \cite{Lin-1} for random attack times, \cite{Lin-2} for imperfect detection, \cite{Hochbaum} on security routing games, and \cite{Basilico-2} for uncertain alarm signals. See also \cite{Baykal} on infracstructure security games. Earlier work on patrolling a channel/border with different paradigms, includes \cite{Washburn-1,Washburn-2}, \cite{Szechtman}, \cite{Zoroa}, and \cite{Collins}. The case of a Patroller restricted to periodic walks has been studied by \cite{Alp-2}. The related problem of ambush is studied by \cite{Baston-1,Baston-2}. An artificial intelligence approach to patrolling is given by \cite{Basilico-1}. Applications to airport security and counter terrorism, are given by \cite{Pita}, and \cite{Fokkink} respectively. Our problem is also akin to that of a restaurant waiter who must return to each table with a certain frequency depending on the number of its diners. An approach to our problem when the locations to be defended have a graph theoretic structure is given in \cite{Alp-3}, where the Patroller can only move along edges. It is shown there that for some graphs the optimal delay is not the value $d=2$ we find here.

Section \ref{Section2} gives a solution to the game for all values of the parameters $m$ and $n$. A closed form solution for the optimal value of $p$ is given for $m=2$ in Section \ref{Section2.1}, for all odd values of $m$ in Section \ref{Section2.2}, and for $m=4$ in Section \ref{Section2.3}. Asymptotic results for $m=4$ as $n$ goes to infinity are given in Section \ref{Section2.4}. Our main result, that $d=2$ and that a location should not be visited in consecutive periods, is established in Theorem \ref{Theorem1} in Section \ref{Section2.5}. In Section \ref{Section2.6} we give an explicit formula for the interception probability as an $m^{th}$ degree polynomial in $p$, and express the probability of intercepting an attack in terms of the distribution of a certain hitting time of a three state birth-death chain. Section \ref{Section2.7} compares the value of the game with the analogous game where the Patroller does not wear a uniform. Section \ref{Section3} concludes.

\section{Analysis}\label{Section2}
\indent

We view the Uniformed Patroller game as being played on a star $S_{n}$ consisting of a center $C$, that is the Patroller's base, connected to $n$ end nodes representing the $n$ locations the Patroller has to defend. We restrict our analysis to a Markovian Patroller, such that she reflects from the $n$ end nodes with probability $s$, from the center $C$ goes to each end node with probability $p$, and remains at the center with probability $r=1-np$. We assume that the two probabilities $p$ and $s$ are both positive so that game is well defined. 
This setting simplifies the game by introducing a two parameter family of patrolling strategies.

\begin{figure}[h!]
\begin{center}
\centering
\psset{unit =0.425cm, linewidth=0.45\pslinewidth}
\begin{pspicture}
\cnodeput(5,10){A}{$E$}\cnodeput(5,0){B}{$E$}\cnodeput(-1,5){C}{$A$}\cnodeput(11,5){D}{$E$}\cnodeput(8.5356,8.5356){E}{$E$}\cnodeput(8.5356,1.4644){F}{$E$}\cnodeput(1.4644,1.4644){G}{$E$}\cnodeput(1.4644,8.5356){H}{$E$}\cnodeput(5,5){I}{$C$}\psset{shortput=nab,arrows=->,labelsep=3pt}\ncarc[arcangle=15]{I}{A}\ncput*{$p$}\ncarc[arcangle=15]{A}{I}\ncput*{$s$}\ncarc[arcangle=15]{I}{D}\ncput*{$p$}\ncarc[arcangle=15]{D}{I}\ncput*{$s$}\ncarc[arcangle=15]{I}{B}\ncput*{$p$}\ncarc[arcangle=15]{B}{I}\ncput*{$s$}\ncarc[arcangle=15]{I}{C}\ncput*{$p$}\ncarc[arcangle=15]{C}{I}\ncput*{$s$}
\ncarc[arcangle=15]{G}{I}\ncput*{$s$}\ncarc[arcangle=15]{I}{G}\ncput*{$p$}
\ncarc[arcangle=15]{F}{I}\ncput*{$s$}\ncarc[arcangle=15]{I}{F}\ncput*{$p$}
\ncarc[arcangle=15]{E}{I}\ncput*{$s$}\ncarc[arcangle=15]{I}{E}\ncput*{$p$}
\ncarc[arcangle=15]{H}{I}\ncput*{$s$}\ncarc[arcangle=15]{I}{H}\ncput*{$p$}
\end{pspicture}
\qquad \qquad \qquad \qquad
\begin{pspicture}
\cnodeput(-2,5){C}{$A$}\put(-4.5,4){,}{}\cnodeput(12,5){D}{$E$}\cnodeput(5,5){I}{$C$}\psset{shortput=nab,arrows=->,labelsep=3pt}\ncarc[arcangle=15]{I}{D}\ncput*{$(n-1)p$}\ncarc[arcangle=15]{D}{I}\ncput*{$s$}\ncput*{$s$}\ncarc[arcangle=15]{C}{I}\ncput*{$s$}\ncarc[arcangle=15]{I}{C}\ncput*{$p$}\ncput*{$p$}
\nccircle[nodesep=4pt]{->}{I}{.5cm}\ncput*{$r$}
\nccircle[nodesep=4pt]{->}{D}{.5cm}\ncput*{$1-s$}
\nccircle[nodesep=4pt]{->}{C}{.5cm}\ncput*{$1-s$}
\end{pspicture}
\end{center}
\caption{\label{Figure4}Left: The Star network $S_{n}$. \label{FigBD}Right: By symmetry, once the attacker has chosen a node $A$, we can reduce the Patroller Markov strategy to a three state birth-death chain on $\{E,C,A\}$. The holding probabilities are shown for clarity on the right, where $r = 1- np$.}
\end{figure}

Let $A$ denote the end node the Attacker has chosen for the attack, and let $E$ denote the set of all other end nodes (with a slight abuse of notation $E$ will sometimes be used to denote an arbitrary end node not equal to $A$). 
Of course, the Patroller does not know which end node is $A$, so this labeling does not affect her patrolling strategy. Note that since we have taken $m\geq 2$ for the attack duration, reflecting from the ends will further imply that the Attacker should never attack at the center since in that case the Patroller will never be away from it for two consecutive periods. First we examine some specific subcases of the game, and then we provide more general results.

\subsection{Short attack duration ($m=2$)}\label{Section2.1}
\indent

We begin our analysis with the fairly simple case when the attack duration $m=2$, since in this case we are able to evaluate explicitly the optimal value of the probability $p$, as well as algebraically prove that the optimal value of probability $s$ is equal to $1$. When the attack duration is $m=2$, an attack at location $A$ cannot be intercepted if it starts when the Patroller is  at another end node, but only if she is at the center. Thus, an attack at $A$ will be intercepted with probability $c\cdot p$, where $c$ is the probability that the Patroller is at $C$ at the beginning of the attack, which in turn implies that the Attacker should choose the waiting time $d$ to minimize $c$. The case $m=2$ has a special property that does not hold for $m>2$, that against $d=2$ the probability $s$ does not affect the interception probability (see Lemma \ref{Lemma1}). Only for $m=2$ are there reflection probabilities $s<1$ which are part of an optimal Patroller strategy. We establish the following result:

\begin{proposition}\label{Proposition1}
Consider the Uniformed Patroller game with arbitrary number $n$ of locations and attack duration $m=2$. An optimal Patroller's strategy is to reflect off the end nodes, $s=1$, and to move from the center with probability 
\begin{equation}\label{eq1}
\hat{p}=1-\frac{\sqrt{n\cdot (n-1)}}{n}\,.
\end{equation}
\noindent The optimal Attacker's strategy is to locate at a random end node and  begin the attack in the second period that the Patroller is away $(d=2)$. The value $V$ of the game is given by 
\begin{equation}\label{eq2}
V=(2\cdot n -1)-2\cdot\sqrt{n\cdot (n-1)}\,.
\end{equation}
\noindent Furthemore $d=2$ is uniquely optimal for the Attacker, while there exists a $k$ less than one such that the strategy $(\hat{p},s)$ is optimal for all $s$ in the interval $[k,1]$.
\end{proposition}

\begin{proof}
We wish to calculate how the probability that the Patroller is at $C$ changes over time, as she continues to be away from the Attacker's attack location $A$. Suppose that at some period $t$ the Patroller is not at $A$, but she is either at $C$ with probability $c$, or belongs to the set $E$ of remaining $n-1$ locations with probability $1-c.$ Then, in the following period $t+1$ the Patroller will be at $C$ with probability $c\cdot r +\left( 1-c\right)\cdot s$, or at $A$ with probability $c\cdot p +(1-c)\cdot 0$. Hence, conditional on the Patroller not being at location $A$ at period $t+1$, the probability that she is at her base $C$ at period $t+1$ is given by
\begin{equation}\label{eq3}
f(c,s)=\frac{c\cdot r+(1-c)\cdot s}{1-p\cdot c}\,.
\end{equation}
\noindent Fraction \eqref{eq3} is increasing in $s$, and $f(c,s)$ is maximized for $s=1$ (we later provide a non-algebraic proof that the Patroller should always adopt $s=1$, see Section \ref{Section2.5}) at
\begin{equation}\label{eq4}
f(c)=f(c,1)=\frac{c\cdot r+(1-c)}{1-p\cdot c}=\frac{1-n\cdot p\cdot c}{1-p\cdot c}\,.
\end{equation}

Additionally, since from equation \eqref{eq4} we find that $f^{'}(c)=-\frac{p\cdot (n-1)}{(1-p\cdot c)^{2}}<0$, then $f(c)$ is strictly decreasing and hence minimized for $c=1$, with $f(1)=\hat{c}$, where
\begin{equation*}
\hat{c}=\frac{(1-n\cdot p)}{(1-p)}\,.
\end{equation*}
\noindent This is easily explained. We have shown that it is optimal for the Patroller to reflect at the ends ($\hat s=1$) so the probability that the Patroller is at the center, given that she was at the center with probability $c$ last period, and has not appeared at the Attacker's end, is simply given by $f(c)$. Since we have also shown that $f(c)$ is decreasing in $c$, and since $c$ is a probability, then $f(c)$ is minimized when $c$ is the maximum probability of being at the center, namely $c=1$. This occurs in the first period that the Patroller has left the Attacker's end. So the minimum probability $\hat{c}$ of the Patroller being at the center occurs two period after she has left the Attacker's end. Hence the optimal delay is $d=2$, regardless of the Patroller's choice of $p$.

An attack at node $A$ with delay $d=2$ will be intercepted if the Patroller is at $C$ in its first period and she goes to $A$ in its second period, that is, with probability
\begin{equation}\label{eq5}
Q=\hat{c}\cdot p=\frac{(1-n\cdot p)\cdot p}{1-p}\,.
\end{equation}

For a number $n$ of candidate attack locations, the Patroller chooses the value of $p$ that maximizes the interception probability \eqref{eq5}. The optimal value $\hat{p}=\hat{p}(n)$ for $p$ depends on $n$ and can be found by solving the first order equation
\begin{equation*}
\frac{\partial Q}{\partial p}=\frac{(n\cdot p^{2}-2\cdot n\cdot p+1)}{(1-p)^{2}}=0\,.
\end{equation*}
\noindent This is equivalent to solving
\begin{equation*}
n\cdot p^{2}-2\cdot n\cdot p +1=0\, ,
\end{equation*}
\noindent giving
\begin{equation}\label{eq6*}
\hat{p}=1-\frac{\sqrt{n\cdot (n-1)}}{n}\,.
\end{equation}

The above findings show that the optimal probability $\hat{p}$ is asymptotic to $\frac{1}{(2\cdot n)}$, while the optimal probability $\hat{r}$ goes to $\frac{1}{2}$. Using \eqref{eq5},\eqref{eq6*}, and the formula for $\hat{c}$, the value $V$ of the game is then given by
\begin{equation*}
V=\hat{c}\cdot\hat{p}=(2\cdot n -1)-2\cdot\sqrt{n\cdot (n-1)} \,.
\end{equation*}

To establish the last part of the Proposition, let $c\left( d\right) =f^{\left( d-1\right) }\left( 1\right)$, with $c(1) =1$, denote the probability that the Patroller is at $C$ after $d$ consecutive periods away from $A$, when she adopts strategy $s=1$ and $p=\hat{p}$.
Hence $c(2) = f(1)=\hat{c}<1$.
We showed that in this case the interception probability is $Q=c\left( d\right) \cdot \hat{p}$ and that $c\left( d\right) $ is minimized when $d=2$ at $V=c\left( 2\right) \cdot \hat{p}=\hat{c}\cdot \hat{p}$.
It is clear that $c\left( d\right) =1$ only for $d=1$ since $p \neq 0$ (i.e. $r\neq 1$).
Also, since the patroller will adopt a strategy with $p = \hat p < 1/n$ the Markov chain is aperiodic and hence $c\left( d\right) $ converges to
some limit $c_{\infty }$ as $d$ goes to infinity (see for example \cite{Norris}). 
We know that $c_{\infty }$ cannot be $1$ because $f\left( 1\right) =c\left( 2\right) <1$, and $f$ is a decreasing function, hence $c_{\infty }<1$ and $c(d) < 1$ for each $d > 1$. 
Since $f$ is strictly decreasing (as shown above), it follows that that $c(d) = f(c(d-1)) > f(1)= \hat{c}$ for $d >2$. Hence, there is some second best value of the delay $d^{\ast }$ for the Attacker, such that for each $d \neq 2$ we have $Q\left( d,s=1\right) \geq c\left( d^{\ast }\right) \cdot \hat{p}\equiv V^{\ast }>\hat{c}\cdot \hat{p} = V$.
By the continuity of $Q$ in $s$
(in the transition probabilities of the Markov chain) and the fact that $Q\left( d=2,s=1\right) =V,$ it follows that there exists a $k$ sufficiently close to $1$ such that $Q\left( d=2,s\right) >V^{\ast }$ for $s\in \left[ k,1\right]$.
Consequently if the Patroller adopts such an $s,$ the best the
Attacker can do is to choose $d=2$ and obtain an interception probability of 
$Q\left( d=2,s\right) =V.$ This means that, as claimed, $\left( s,\hat{p}%
\right) $ is optimal for the Patroller for $s\in \left[ k,1\right]
.\allowbreak $
\end{proof}

The fact that for $m = 2$ there are optimal Patroller strategies with $s < 1$ is in stark contrast to the unique optimality of $s = 1$ for larger m which we establish in Lemma \ref{Lemma1}.

\subsection{Odd attack duration ($m$ odd)}\label{Section2.2}
\indent

When the attack duration $m$ is odd, we establish the following result:

\begin{proposition}\label{Proposition2}
Consider the Uniformed Patroller game where the attack duration $m$ is odd. The unique optimal Patroller's strategy is the random walk with zero holding probabilities, $\hat p=\frac{1}{n}$, $\hat s=1$. The optimal Attacker's strategy is to attack at any end with an arbitrary delay. The value of the game is then given by
\begin{equation}\label{eq6}
V=1-\left( \frac{n-1}{n}\right) ^{(m-1)/2}\,.
\end{equation}
\end{proposition}

\begin{proof}
Say the Patroller adopts a random walk with zero holding probabilities on the star $S_n$, such that she always goes to a uniform random adjacent node, and never remains at the same node. Say the attack duration is $m = 2\cdot j + 1$. In any $2\cdot j$ consecutive periods, the Patroller visits $j$ ends counting multiplicity, chosen uniformly at random. The probability that a particular end is visited among these $j$ ends is given by $1-(\frac{n-1}{n})^j$, and so this is a lower bound on the interception probability. Similarly, if the Attacker begins his attack at any time the Patroller is away from his chosen location, then the number $k$ of end nodes the Patroller visits during the rest of the attack satisfies $k \leq j$, and since these are chosen randomly (independently) the probability that the Attacker’s node is among them is given by 
\begin{equation*}
1-\Big(\frac{n-1}{n}\Big)^{k} \leq 1-\Big(\frac{n-1}{n}\Big)^j\,.
\end{equation*} 
\noindent However, since only the random walk with zero holding probabilities $p=1/n$, $s=1$ always gives $k = j$, it follows this is the unique optimal Patroller’s strategy and the value of the game is given then by \eqref{eq6}.
\end{proof}

\subsection{Attack duration $m=4$}\label{Section2.3}
\indent

We can also obtain an explicit solution to the game when $m=4$ (but not for higher even $m$). When the attack duration is $m=4$, if the attack starts when the Patroller is at the center, then she gets two chances to intercept it (i.e. to find the correct end $A$). However, if the attack starts when the Patroller is at an end $E\neq A$, then she gets only one chance. In the first case, the Patroller can intercept the attack with any of the following four sequences 
\begin{equation*}
CA\_\, \_ ,\quad CCA\_,\quad CCCA,\quad CECA,
\end{equation*}
\noindent while in the second case, with any of the following three sequences
\begin{equation*}
ECA\_,\quad ECCA,\quad EECA.
\end{equation*}

It follows that the attack will be intercepted with overall probability
\begin{gather}
\begin{aligned}\label{eq7}
Q&=c\cdot \bigl(1+r+r^{2}+( 1-p-r)\bigr)\cdot p+( 1-c)\cdot\bigl( s+s\cdot r+( 1-s)\cdot s\bigr)\cdot p
 \\&= p\cdot ( -r\cdot s+r^{2}+s^{2}-2\cdot s-p+2)\cdot c+p\cdot (
2\cdot s+r\cdot s-s^{2})\,,
\end{aligned}
\end{gather}
\noindent where $c$ is the conditional probability that the Patroller is at the center $C$ at the beginning of the attack given that she is not at the attack node $A$.

Note that the coefficient of $c$ in \eqref{eq7} is given by the product of $p$ with the expression
\begin{gather*}
\begin{aligned}
( -r\cdot s+r^{2}+s^{2}-2\cdot s-p+2)&=2+( s^{2}-2\cdot s)-r\cdot s+r^{2}-p \\
\geq 2-1+-rs+r^{2}-p &\geq 1-\left( r+p\right) +r^{2}  \geq 1-\left( r+p\right) 
\geq 0\,,
\end{aligned}
\end{gather*}
\noindent since $r+p\leq r+n\cdot p=1$, and $s\leq 1$. It follows that for fixed $n$, $p$, the interception probability \eqref{eq7} is increasing in $c$. By the same reasoning we have used for $m=2$, it further follows that the Attacker should choose to wait for $d=2$ to attain $c=\hat{c}$, and the minimum interception probability
\begin{equation}\label{eq8}
Q=\frac{p}{1-p}\bigl(( p+r-1)\cdot s^{2}+(2-r-p\cdot r-r^{2}-2\cdot p)\cdot s+( 2\cdot r-p\cdot r+r^{3})\bigr)\,.
\end{equation}

To check that \eqref{eq8} is increasing in $s$, notice that the derivative with respect to $s$ in the bracketed quadratic above is given by
\begin{equation*}
2\cdot ( p+r-1)\cdot s+( 2-r-p\cdot r-r^{2}-2\cdot p) \geq 2\cdot ( p+r-1)\cdot 1+( 2-r-p\cdot r-r^{2}-2\cdot p)\,,
\end{equation*}
\noindent since $r+p<1$, which is equivalent to $r\cdot (1- p-r) \geq 0,$ since both factors are non-negative.

Consequently, it turns out that the Patroller maximizes the interception probability for fixed $p$ by choosing $s=1$ (reflecting at the ends). Taking $s=1$ in \eqref{eq8}, gives
\begin{gather}\label{eq9}
\begin{aligned}
Q&=\frac{p}{1-p}\cdot \bigl( 2\cdot r-p-2\cdot p\cdot r-r^{2}+r^{3}+1\bigr) 
 \\
&=\frac{-n^{3}\cdot p^{4}+2\cdot n^{2}\cdot p^{3}+2\cdot n\cdot
p^{3}-3\cdot n\cdot p^{2}-3\cdot p^{2}+3\cdot p}{1-p}\,.
\end{aligned}
\end{gather}

We can also get an exact algebraic expression for the value $p=\hat{p}$ which maximizes the interception probability \eqref{eq9}, by differentiating \eqref{eq9} with respect to $p$ and setting the numerator of the resulting fraction $\frac{f(p)}{(1-p)^{2}}$ equal to $0$, where $f(p)$ is the fourth degree polynomial
\begin{equation*}
3-(6+6\cdot n)\cdot p +(3+9\cdot n+6\cdot n^2)\cdot p^2+(-4\cdot n-4\cdot n^2+4\cdot n^3)\cdot p^3+(3\cdot n^3)^4\,.
\end{equation*}

\noindent The real root of this polynomial which is a probability, is given by
\begin{equation}\label{eq10*}
\hat{p}=\frac{2\cdot n^2+\frac{G}{\sqrt{2}}-3\cdot n^2\cdot\sqrt{\frac{C}{9\cdot n^4}+\frac{E}{3\cdot n^3\cdot D}+\frac{4\cdot \sqrt{2}\cdot F}{9\cdot n^2\cdot G}-\frac{D}{3\cdot n^3}}+2\cdot n+2}{6\cdot n^2}\,,
\end{equation}
\noindent where
\begin{equation*}
\setlength{\jot}{15pt}
\begin{split}
A&=8\cdot n^6-18\cdot n^5+6\cdot n^4+9\cdot n^3-3\cdot n, \\
B&=(n-1)^6\cdot n^3\cdot (32\cdot n^3+24\cdot n^2-3\cdot n-4),\\
C&=8\cdot (n^2+n+1)^2-12\cdot n\cdot (2\cdot n^2+3\cdot n+1), \\
F&=(4\cdot n^4+2\cdot n^3+6\cdot n^2+11\cdot n+4)\cdot(n-1)^2,
\end{split}
\qquad
\begin{split}
D&=\sqrt[3]{A+2\cdot\sqrt{B}-3\cdot n+1},\\ E&=(4\cdot n^2-1)\cdot (n-1)^2, \\ G&=\sqrt{C-\frac{6\cdot n\cdot E}{D}+6\cdot n\cdot D}\,.
\end{split}
\end{equation*}

\subsection{Asymptotic analysis for $m=4$}\label{Section2.4}
\indent 

We now consider the optimal play for the attack duration $m=4$ when $n$ is large. Since in this case the probability $p$ goes to $0$, it is more transparent to work with probability $r$ of remaining at the center. We start by writing \eqref{eq9} in terms of $r$, using that $p=\frac{(1-r)}{n}$, and calling it
\begin{equation*}
\pi(n,r)=\frac{(1-r)\cdot \bigl(-1-r+2\cdot r^2+n\cdot(1+2\cdot r-r^2+r^3)\bigr)}{n\cdot(-1+n+r)}\,,
\end{equation*}
\noindent so that,
\begin{equation*}
n\cdot \pi(n,r)\rightarrow poly(r)=-r^4+2\cdot r^3-3\cdot r^2+r+1,\quad \text{as}\quad n\rightarrow \infty\,.
\end{equation*}
\noindent The first order condition on $r$ is given by the cubic
$4\cdot r^3-6\cdot r^2+6\cdot r-1=0$, with solution in $[0,1]$ of
\begin{equation}\label{11*}
\begin{split}
\hat{r}_{\infty}&=\frac{1}{2}\Bigl(1-\bigl(-1+\sqrt{2}\bigr)^{-1/3}+\bigl(-1+\sqrt{2}\bigr)^{1/3}\Bigr)\simeq 0.20196\,.
\end{split}
\end{equation}

We also have that
\begin{equation*}
\begin{split}
a\equiv poly(\hat{r}_{\infty})&=-\frac{3\cdot \Bigl(5-4\cdot \sqrt{2}-7\cdot\bigl(-1+\sqrt{2}\bigr)^{4/3}+\bigl(-1+\sqrt{2}\bigr)^{2/3}\cdot\bigl(-1+2\cdot \sqrt{2}\bigr)\Bigr)}{16\cdot \bigl(-1+\sqrt{2}\bigr)^{4/3}}\simeq 1.0944\,,
\end{split}
\end{equation*}
\noindent which implies that
\begin{equation*}
\pi(n,\hat{r}_{n})\rightarrow \frac{a}{n} \simeq \frac{1.0944}{n}\,.
\end{equation*}

Hence, we show that for large $n$ the Patroller should stay at the center about $20\%$ of the time, which ensures her an interception probability of about $\frac{1.09}{n}$. The optimal delay remains~$d=2$.

\subsection{Optimality of $s=1$ and $d=2$ for all $m,n$}\label{Section2.5}
\indent

Let us now extend our analysis to even $m\geq 6$. Firstly, it is useful to show that an attack delay of $d=2$ and a reflection probability of $s=1$ are optimal for the Attacker and the Patroller respectively. This result (Theorem \ref{Theorem1}) was conjectured by the Associate Editor.

Let $X_{i} \in \{C,E,A\}, i=1,2,\dots$, denote the location of the Patroller in the $i^{th}$ time period (where $E$ denotes the set of end nodes other than the attack node $A$). 
Then the one sided sequence $X=(X_{1},X_{2},\dots)$ is a sample path of the Patroller's Markovian strategy.
Let $T(X)=\min\{i:X_{i}=A\}$ denote the first time that the Patroller reaches the attack node $A$, this random time is typically called the hitting time of $A$.

\begin{lemma}\label{Lemma1}
Assume the Attacker adopts a delay of $d=2$. Then, for any $n,m,p$ the interception probability is maximised when $s=1$. That is, reflection at ends is an optimal response for the Patroller. Furthermore, if $m\geq 3$ any $s<1$ is not an optimal response, but for $m=2$ some values of $s<1$ are optimal.
\end{lemma}

\begin{proof}
The proof uses a technique which in the probability literature is known as a Markov coupling, that is we consider simultaneously two Markov sample paths constructed using the same source of randomness. Consider any Patroller's Markov chain (patrol) strategy $(p,s)$. 
Set period $i=1$ to be the period that the Attacker begins his attack at some node we denote $A$. 
At time $i=1$ every Patroller sample path is either at $X_{1}=C$ (center node) or $X_{1}=E$ (any end node other than $A$). 
Since the Attacker adopts a delay of $d=2$, and the Patroller will surely be at $C$ in the first period away from $A$, and the Attacker will start their attack if an only if the Patroller is at either $C$ or $E$ at the next step as well, the probability that the Patroller will be at $C$ again at $i=1$ is given by the probability that the Patroller walk stays at $C$ conditioned on not moving $A$, i.e.  $Pr(X_{1}=C)=\frac{1-n\cdot p}{1-p}$, which does not depend on $s$.
%
%
Note that the Attack at $A$ beginning at time $i=1$ is intercepted if and only if $T(X) \leq m$. For any sample path $X$, let $\widetilde{X}$ be the same sequence with all consecutive appearances of $E$ deleted (after the first one). So for example, if $X=(C,E,E,E,C,C,E,C,A,\dots)$ then we have $T(X)=9$, and $\tilde{X}=(C,E,C,C,E,C,A,\dots)$ with $\widetilde{T}(X)\equiv T(\widetilde{X})=7$. 
Recall that although we lump all ends other than $A$ together, consecutive elements of $E$ visited by the Patroller must be the same node. 
It follows from this construction that the paths $\widetilde{X}$ are sample paths on $\{C,E,A\}$ with distribution that arise from a Patroller with $s=1$, i.e. with strategy $(p,1)$.
It is clear that we always have $\widetilde{T}(X) \leq T(X)$, hence if $T(X) \leq m$ we also have $\widetilde{T}(X) \leq m$. 
Since $\widetilde{T} \leq T $, and we have not changed the value of $p$ for both sample paths $X$ and $\widetilde{X}$, the Markov chain $(p,1)$ is at least as good for the Patroller as $(p,s)$.

To establish the last part of the Lemma, suppose that $m\geq 3.$ Let $K
$ be the set of sample paths of $(p,s) $
starting with $$\bigl(\overbrace{E,...,E}^{m-1},C,A\bigr)\,,$$ 
so that $T\left( X\right) =m+1$ and $\widetilde{T}\left( x\right) =3\leq m$ for $X\in K.$
Note that  $\Pr \left( K\right) =\alpha _{s}>0$ for $s<1$.
It follows, from the coupling described above, that  $Q(s=1)\geq Q(s)+\alpha _{s}>Q(s)$ for $s<1$, and so any Patroller
strategy with $s<1$ cannot be an optimal response to $d=2$. 
Notice that when $m=2$ it is not possible to construct an $s < 1$ path that fails to intercept the attacker but the coupled version with $s=1$ does intercept the attacker. Hence the proof strategy does not work in this case. We showed in Proposition \ref{Proposition1} that for $m=2$ there are values of $s<1$ which are optimal.
\end{proof}

We now consider a simple question. Is it better to begin an attack when the Patroller is at another end node $E$ or when she is at the center $C$? We show that the attack should begin when the probability of the former is higher.

\begin{lemma}\label{Lemma2}
For any Patroller strategy $(p,s)$ and any $m,n$, the interception probability if the attack starts when $X_{1}=C$ is at least as high as if $X_{1}=E$. That is
\begin{equation}\label{eq10}
Pr(T(X) \leq m \mid X_{1}=C) \geq Pr(T(X) \leq m \mid X_{1}= E)\,.
\end{equation}
\noindent Furthermore, the left and right sides of \eqref{eq10} are equal if and only if $s=1$, $r=0$ and $m$ is odd. In particular, for even $m$ it is strictly better to start the attack when the Patroller is at the center.
\end{lemma}

\begin{proof}
A Patroller at an end $E$ must go to the center $C$ before reaching the attacked node $A$. It follows that the probability of reaching $A$ from $E$ in $t$ periods is no greater than the probability of reaching $A$ from $C$ in $t-1$ periods. Hence, by the Markov property, 
\begin{equation}\label{eq11}
Pr(T(X) \leq m \mid X_{1}=E) \leq Pr(T(X) \leq m-1 \mid X_{1}=C) \leq Pr(T(X) \leq m \mid X_{1}=C)\,,
\end{equation}
\noindent because if the Patroller reaches $A$ in fewer than $m-1$ periods, then she also reaches it in fewer than $m$ periods. 

Regarding the `furthermore' part, if the left and right hand sides of \eqref{eq12} are equal, this means that both the left and right `$\leq$' signs hold with equality. If the left `$\leq$' holds with equality, we must have $s=1$ because $Pr(X_2=C|X_1=E)=1$. Now suppose the right `$\leq$' sign holds with equality, so that for $X_1=C$ we have
\begin{equation*}
Pr(T(x)\leq m-1)=Pr(T(x)\leq m)\,.
\end{equation*}
\noindent It follows that 
\begin{equation*}
Pr(T(x)=m)=0\,.
\end{equation*}
\noindent This implies that $r=0$ (the chain is periodic) because otherwise this probability is at least $r^{m-1}\cdot \frac{1-r}{n}$. When $r=0$ and $s=1$ the Patroller follows a random walk (with zero probability of remaining at any node) which will be at the center in odd periods and at an end at even periods. Since $T(x)\neq m$ she is at the center at period $m$, so $m$ must be odd.

\end{proof}

It is important to note that if the attack duration $m$ is odd, and $p=\frac{1}{n}$, then the two probabilities given in \eqref{eq10} are strictly equal. To check this, suppose we label the $m$ periods of the attack as $t=1,2,\dots,m=2k+1$. If the Patroller is at the center $C$ at the beginning $t=1$ of the attack, then she will be at a random end node (either the one that is attacked or another one) in the following $k$ even numbered periods $2,4,\dots, m-1=2k$. If, however, the Patroller is at an end node $E$ (other than $A$) at the beginning $t=1$ of the attack, then she will be at a random end node in the following $k$ odd numbered periods $3,5,\dots, m=2k+1$. In either case, the interception probability is given by $1-(\frac{n-1}{n})^k$, as seen in the proof of Proposition \ref{Proposition2}. This fact explains why the delay $d=2$, which maximizes the probability of being at the center $C$ in the first period of the attack, is not required for optimality when $m$ is odd.

\begin{lemma}\label{Lemma3}
Suppose the Patroller is adopting a strategy $(p,s)$ with $s=1$. 
Then, for any $n,m$ the interception probability is minimised when $d=2$. Furthermore, if $m$ is odd then $d=2$ is the Attacker's unique optimal response.
\end{lemma}

\begin{proof}
We will show that the probability that the Patroller is at $C$ given that she has not returned to $A$ (the attack node) when the attack begins is minimised when $d=2$. Hence, by Lemma \ref{Lemma2}, this is an optimal delay.

The proof follows exactly the proof that $d=2$ is optimal in Proposition \ref{Proposition1}, which is independent of $m$. Recall the probability $f(c)$ that the Patroller is at the center, given that she was at the center the last period with probability $c$ and has not returned the current period to the attack node $A$, is given by
\begin{equation}\label{eq12}
f(c)=\frac{c\cdot r+(1-c)}{1-c\cdot p}\,.
\end{equation}

This is easily seen to be decreasing in $c$.
Since $c=1$ (its maximum possible value) at the first period that the Patroller is away from $A$, and $f$ is decreasing, the minimum value of $f$ is attained as $f(1)=\frac{1-n\cdot p}{1-p}$ after delay $d=2$. By Lemma \ref{Lemma2}, the delay $d=2$ is the optimal response to any Patroller's strategy $p$ when $s=1$.
\end{proof}

According to Lemmas \ref{Lemma1} and \ref{Lemma3}, an attack delay of $d=2$ and a Patroller reflection probability of $s=1$ are optimal responses to each other. Thus, we have established the following result:

\begin{theorem}\label{Theorem1}
For any attack duration $m$ and any number $n$ of locations, the Attacker's strategy with delay $d=2$ and the Patroller's strategy with reflection probability $s=1$ are optimal strategies. Furthermore, the delay $d=2$ is uniquely optimal if and only if $m$ is even, and the reflection probability $s=1$ is uniquely optimal if and only if $m\geq 3$. 
\end{theorem}

\begin{proof}
Lemmas \ref{Lemma1} and \ref{Lemma3} establish that this pair of strategies forms a Nash equilibrium, with a corresponding payoff given by the interception probability $\hat{Q}$. Since this is a zero-sum game, the strategies in any Nash equilibrium are optimal and the payoffs they produce is the value $V=\hat{Q}$ of the game. That is, playing the claimed optimal strategies ensures that a player gets an interception probability (payoff) at least as good as $\hat{Q}$. The `furthermore' part follows from the `furthermore' parts of Lemmas \ref{Lemma1} and \ref{Lemma3} and Proposition \ref{Proposition2}.
\end{proof}

Theorem \ref{Theorem1} was conjectured by the anonymous Associate Editor, who also suggested a proof technique somewhat different from the one given above.

\subsection{A general formula for the interception probability $Q$}\label{Section2.6}
\indent

In Theorem \ref{Theorem1} we showed that the optimal attack delay $d$ for the Attacker and the optimal reflecting from the end nodes probability $s$ for the Patroller are $d=2$ and $s=1$ respectively. 
Assuming that the Patroller and Attacker adopt these values we are able to derive an explicit formula, in terms of probability $p$, for the interception probability $Q$, for any attack duration $m$ and any star network size $n$. We obtain this formula either by a direct analysis of the hitting time, $T$, through probability generating functions, or by a recursive analysis on $m$. We also note that using the former method it is possible to derive an explicit expression for $Q$ for all values of $d$ and $s$ as well (in complete generality), but it does not appear to be possible to perform explicit algebraic optimisation, and so we only present the analysis in the simpler setting of $d=2$ and $s=1$ for the sake of clarity of the presentation.

\begin{proposition}\label{Proposition3}
Suppose the players adopt the strategies with delay $d=2$ and reflection probability $s=1$ (shown to be optimal in Theorem \ref{Theorem1}). Then the probability that an attack of length $m$ is intercepted on the star network $S_n$ when the Patroller goes to each node from the center with probability $p$, is given by
\begin{equation}\label{eq15*}
Q=1-\frac{(1-2\cdot p+n\cdot p+u)\cdot w_{2}^{m}-(1-2\cdot p+n\cdot p-u)\cdot w_{1}^{m}}{2\cdot (1-p)\cdot u},
\end{equation}
\noindent where $u=\sqrt{(n\cdot p +1)^2 -4\cdot p}$, $w_1=\frac{1-n\cdot p-u}{2}$ and $w_2=\frac{1-n\cdot p+u}{2}$\,.
\end{proposition}

We give two proofs of Proposition \ref{Proposition3}. 
The first one involves a direct calculation based on the hitting times of birth-death chains, using generating functions and applies more generally. 
The second approach involves a recursion on the attack duration $m$ and relies strongly on the fact we take $d=2$ and $s=1$.
The first approach makes use of the fact that the distribution of hitting times for birth-death chains can be expressed explicitly in terms of their probability generating function, see for example \cite{Fill}.
Birth-death chains are Markov chains on the natural numbers which can increase or decrease in each time step by at most one - the Patrollers Markov strategy on the reduced state space $\{C,E,A\}$ is therefore a special case of a birth-death chain.

\noindent \textit{First Proof of Proposition \ref{Proposition3} (based on hitting time analysis).}  In the first period that the Patroller leaves $A$ the Attacker starts to count their delay, and the Patroller must be at the centre $C$. 
Let $T_C(X)=\min\{i\,;\,X_i=A\}$ denote the first hitting time of the attacked node $A$ for the Patroller Markov strategy $X$ on the reduced space $\{C,E,A\}$, started from the centre node $C$.
If the Patroller returns to $A$ fewer than $d$ periods later then the attack will not begin and the Attacker resets their delay counter.
On the other hand, if the Patroller is away from $A$ for $d$ or more periods then the attack will begin, and it will be intercepted if and only if the Patroller returns to $A$ strictly before time period $d+m -1$ (the minus one appear because we consider that the attack has begins on the $d^{th}$ period that the Patroller is away).
It follows that the interception probability $Q$ can be expressed as
\begin{equation}\label{eq15}
\begin{split}
Q&=Pr(T_{C} < d+m-1 \mid T_{C} \geq d)\\&=1-Pr(T_{C}\geq d+m-1 \mid T_{C}\geq d) = 1- \frac{Pr(T_{C}\geq d+m-1)}{Pr(T_{C}\geq d)}\,.
\end{split}
\end{equation}
\noindent Hence, if we know the distribution of the random time $T_C$ then it is possible to calculate the probability $Q$.
It turns out that for birth-death chains the probability generating function of the hitting times (which uniquely determines the distribution) can be written explicitly in terms of eigenvalues of the Markov transition matrix. This result is summarised in the following result taken from \cite{Fill}.

\bigskip Theorem 1.2 in \cite{Fill}.
\textit{Consider a discrete-time birth-death chain with transition matrix $P$ on state space $\{0,\ldots,N\}$ started at $0$, with positive birth probabilities $p_i$, $0\leq i \leq N-1$ and death probabilities $q_i$, $1\leq i\leq N-1$. The the first hitting time, $\tau$, of state $N$ has probability generating function}
\begin{align*}
    G_\tau(z) =\mathbb{E}[z^\tau]= \sum_{t=1}^{\infty}z^t Pr(\tau = t) = \prod_{j=0}^{N-1} \frac{(1-\theta_j)z}{1-\theta_j z}\,,
\end{align*}
\textit{where $-1\leq \theta_i<1$ are the $N$ non-unit eigenvalues of $P$.}

\bigskip To apply this theorem in our setting we identify the set, $E$, with $0$, the centre node, $C$, with $1$ and the attacked node, $A$, with $N=2$ (see Figure \ref{FigBD} right).
Then, denoting by $T_E(X)=\min\{i\,;\,X_i=A\}$ the first hitting time of the attacked node $A$ for the Patroller Markov strategy $X$ started from $E$,  by Fill's Theorem  the  probability generating function of $T_E$ is given by
\begin{align}
    \label{eq:TEgen}
    G_{T_E}(z) = \frac{(1-\theta_{-})\cdot z}{1-\theta_{-}\cdot z} \cdot \frac{(1-\theta_{+})\cdot z}{1-\theta_{+}\cdot z},
\end{align}
where $-1\leq \theta_{-}$, $\theta_{+} < 1$ are the non-unit eigenvalues of the Markov transition matrix 
\begin{equation}\label{eq17}
  P = \bordermatrix{~  &  E & C & A
                       \cr
                    E & 1-s & s & 0  \cr
                    C & (n-1)p & 1-np & p \cr
                    A & 0 & 0 & 1\cr
                    }\,.
\end{equation}
\noindent The non-unit eigenvalues of the transition matrix \eqref{eq17} are given by
\begin{equation*}
\theta_{\pm}=\frac{1}{2}\cdot \bigl(2-s-n\cdot p\pm \sqrt{(s+n\cdot p)^{2}-4\cdot p\cdot s}\bigr)\,.
\end{equation*}

Finally, in order to calculate the generating function for the hitting time $T_C$ (the hitting time of $A$ started from $C$), we observe that for the Patroller to reach $A$ from $E$ it must first reach the state $C$.
The Patroller started from $E$ will first take an amount of time which is geometrically distributed with parameter $s$ before it reaches $C$, and then it will take a further (independent) random time $T_C$ before reaching $A$.
It follows, by the strong Markov property, that we have the following equality in distribution;
\begin{equation*}
T_{E}=S+T_{C}\,,
\end{equation*}
\noindent where $S$ is a Geometric$(s)$ random variable independent of $T_C$.
Since $S$ is a geometric random variable it's probability generating function is given by
\begin{align*}
    G_{S}(z) = \frac{s\cdot z}{1-(1-s)\cdot z}\,.
\end{align*}
\noindent The probability generating function of the sum of independent random variables is given by the product of the generating functions, and hence we have
\begin{align*}
    G_{T_C}(z) = \frac{G_{T_E}(z)}{G_S(z)} = \Bigl(\frac{1-(1-s)\cdot z}{s\cdot z}\Bigr)\cdot \Bigl(\frac{(1-\theta_{-})\cdot z}{1-\theta_{-}\cdot z}\Bigr) \cdot \Bigl(\frac{(1-\theta_{+})\cdot z}{1-\theta_{+}\cdot z}\Bigr) \,.
\end{align*}
\noindent Under the assumptions $d=2$ and $s=1$, which have been shown to be optimal in Theorem \ref{Theorem1}, the expression above can be simplified into the form
\begin{align}\label{eq21*}
G_{T_{C}}(z)&=\cdot\frac{4\cdot p \cdot z}{\left(2-(1-n\cdot p-u)\cdot z\right)\cdot\left(2-(1-n\cdot p+u)\cdot z\right)}\nonumber\\
&=\frac{p\cdot z}{1-z\cdot (1-n\cdot p)-z^2\cdot (n\cdot p -p)}\,,
\end{align}
\noindent where we recall $u =\sqrt{(n\cdot p +1)^2-4\cdot p}$.
The probability mass function of the discrete random variable $T_{C}$ can then be calculated in the standard way from the probability generating function, $Pr(T_{C}=k)=G^{(k)}_{T_{C}}(0)/k!$ where $G^{(k)}_{T_{C,A}}(0)$ is the $k^{th}$ derivative of $G_{T_C}(z)$, evaluated at $z=0$.
Taking derivatives in \eqref{eq21*} we have
\begin{equation}
\label{eq:PrT}
Pr(T_{C}=k)=\frac{2^{-k}\cdot p\cdot \bigl((1-n\cdot p+u)^k-(1-n\cdot p-u)^k\bigr)}{u}=\frac{p\cdot (w_{2}^{m}-w_{1}^{m})}{u}\,,
\end{equation}
\noindent where  $w_1=\frac{1-n\cdot p-u}{2}$ and $w_2=\frac{1-n\cdot p+u}{2}$.

Finally plugging the form of the probability mass function for the hitting time $T_C$ give in \eqref{eq:PrT} into the formulation of $Q$ in terms of  $T_C$, in \eqref{eq15}, we get an explicit expression of the interception probability $Q$ as a function of $m,n,p$ (for $s=1$,~$d=2$)
\begin{equation*}
\begin{split}
Q&=1-\frac{Pr(T_{C,A} \geq m+1)}{Pr(T_{C,A} \geq 2)}\\&=1-\frac{1-\bigl(Pr(T_{C,A}=1)+Pr(T_{C,A}=2)+\dots+Pr(T_{C,A}=m)\bigr)}{1-Pr(T_{C,A}=1)}\\
&=1-\frac{1-\bigl(\frac{G^{(1)}_{T_{C,A}}(0)}{1!}+\frac{G_{T_{C,A}}^{(2)}(0)}{2!}+\dots+\frac{G_{T_{C,A}}^{(m)}(0)}{m!}\bigr)}{1-\frac{G^{(1)}_{T_{C,A}}(0)}{1!}}=1-\frac{1-\frac{p}{u}\cdot \sum_{n=1}^{m}(w_{2}^n-w_{1}^n)}{1-p}\\
&=1-\frac{(1-2\cdot p+n\cdot p+u)\cdot w_{2}^{m}-(1-2\cdot p+n\cdot p-u)\cdot w_{1}^{m}}{2\cdot (1-p)\cdot u}\,.\pushQED{\qed}\qedhere
\end{split}
\end{equation*}

\noindent \textit{Second Proof of Proposition \ref{Proposition3} (based on a recursion on $m$).}
In this proof we insist that $s=1$ and $d=2$ from the beginning.
Let $h(m)$ denote the probability that a Patroller, who is at the center $C$ in the first period of an attack lasting for $m$ periods, does not intercept the attack (i.e. the probability the attack is successful). Recall $p$, $q=(n-1)\cdot p$ and $r=1-n\cdot p$ denote the probabilities of the Patroller going from the center $C$ to $A$, $E$ and $C$ respectively, where $E$ stands for the set of end nodes other than the attack node $A$. 
If the attack lasts for $1$ period, then it will clearly not be intercepted, so $h(1)=1$. 
If, however, the attack lasts for $2$ periods, then it will be successful only if the Patroller stays at the center or goes to another node $E$, so we have $h(2)=r+q$. 
For general attack length $m$, if the Patroller stays at the center the attack will be successful with probability $h(m-1)$, and if she goes to an end $E$ (after which she always comes back to $C$ since $s=1$) the attack will succeed with probability $h(m-2)$, as this return trip takes $2$ periods. Thus, we have a recursion for $h(m)$ consisting of 
\begin{equation}\label{eq23}
h(m)=r\cdot h(m-1)+q\cdot h(m-2)\,,
\end{equation}
\noindent with initial conditions $h(1)=1$, $h(2)=q+r$.

According to the elementary theory of recursions, see e.g. Lecture $2$ of \cite{Brousseau}, the solution for $h$ is given by the function 
\begin{equation}
h(m)=a\cdot w_{1}^{m}+b\cdot w_{2}^{m}\,,
\end{equation}
\noindent for some $a,b$, where $w_{1},w_{2}$ with $w_{1}<w_{2}$, are the two distinct roots of the characteristic equation of \eqref{eq23}, that is
\begin{equation*}
z^2-r\cdot z-q=0\,.
\end{equation*}

Solving the simultaneous equations 
\begin{equation*}
\begin{split}
h(1)&=a\cdot w_{1}+b\cdot w_{2}=1\quad,\quad
h(2)=a\cdot w_{1}^{2}+b\cdot w_{2}^{2}=1-p\,,
\end{split}
\end{equation*}
\noindent we evaluate the constants $a$ and $b$, as
\begin{equation*}
a=-\frac{(1-2\cdot p+n\cdot p-u)}{u\cdot (1-n\cdot p-u)}\quad,\quad b=\frac{(1-2\cdot p+n\cdot p+u)}{u\cdot (1-n\cdot p+u)}\,.
\end{equation*}

Now we relax the assumption that the Patroller is at $C$ in the first period of the attack. If the Patroller is at end node $E$($\neq A$) at the start of an attack of length $m$, then she will be at $C$ after one more period, which is equivalent to being at $C$ in the first period of an attack of length $m-1$. Thus, in this case the attack will succeed with probability $h(m-1)$, by the definition of $h$. Since the attack will start after a delay $d=2$, the probability that the attacker is at the center in the first period of the attack is given by $c=\frac{(1-n\cdot p)}{(1-p)}$. It follows that the overall probability $1-Q$ of a successful attack of length $m$ is given by the formula
\begin{equation}
\begin{split}
1-Q&=c\cdot h(m)+(1-c)\cdot h(m-1)\\
&=c\cdot (a\cdot w_{1}^{m}+b\cdot w_{2}^{m})+(1-c)\cdot (a\cdot w_{1}^{m-1}+b\cdot w_{2}^{m-1})\\
&=\bigl(c\cdot a+\frac{(1-c)\cdot a}{w_{1}}\bigr)\cdot w_{1}^{m}+\bigl(c\cdot b+\frac{(1-c)\cdot b}{w_{2}}\bigr)\cdot w_{2}^{m}\,.
\end{split}
\end{equation}

We need to show that the coefficients of $w_{1}^{m}$ and $w_{2}^{m}$ are as given in \eqref{eq15*}.
\begin{equation*}
\begin{split}
c\cdot a+\frac{(1-c)\cdot a}{w_1}&=
-\frac{1-n\cdot p}{1-p}\cdot \frac{1-2\cdot p+n\cdot p-u}{u\cdot (1-n\cdot p-u)}-\frac{n\cdot p-p}{1-p}\cdot \frac{1-2\cdot p+n\cdot p-u}{u\cdot (1-n\cdot p-u)}\cdot\frac{2}{1-n\cdot p-u}\\
&=\frac{-(1-2\cdot p+n\cdot p-u)}{2\cdot (1-p)\cdot u}\cdot \frac{2\cdot (1-u+n^2\cdot p^2+n\cdot p\cdot u-2\cdot p)}{(1-n\cdot p-u)^2}\\
&=\frac{-(1-2\cdot p+n\cdot p-u)}{2\cdot (1-p)\cdot u}\cdot \frac{2\cdot (1-u+n^2\cdot p^2+n\cdot p\cdot u-2\cdot p)}{1-2\cdot n\cdot p+2\cdot n\cdot p\cdot u-2\cdot u+n^2\cdot p^2+u^2}\\
&=\frac{-(1-2\cdot p+n\cdot p-u)}{2\cdot (1-p)\cdot u}\,,
\end{split}
\end{equation*}
\noindent since the second factor in the second last line is seen to be equal to $1$ after making the substitution $u^2=n^2\cdot p^2+2\cdot n\cdot p-4\cdot p+1$.\qed

In Figure \ref{Figure2} we plot the interception probability $Q$ as a function of $p\in [0,1/n]$, for even $m=2,4,8,10$ (black) and for odd $m=3,5,7,9$ (green) attack duration $m$.

\begin{figure}[H]
\begin{center}
    \includegraphics[scale=0.4]{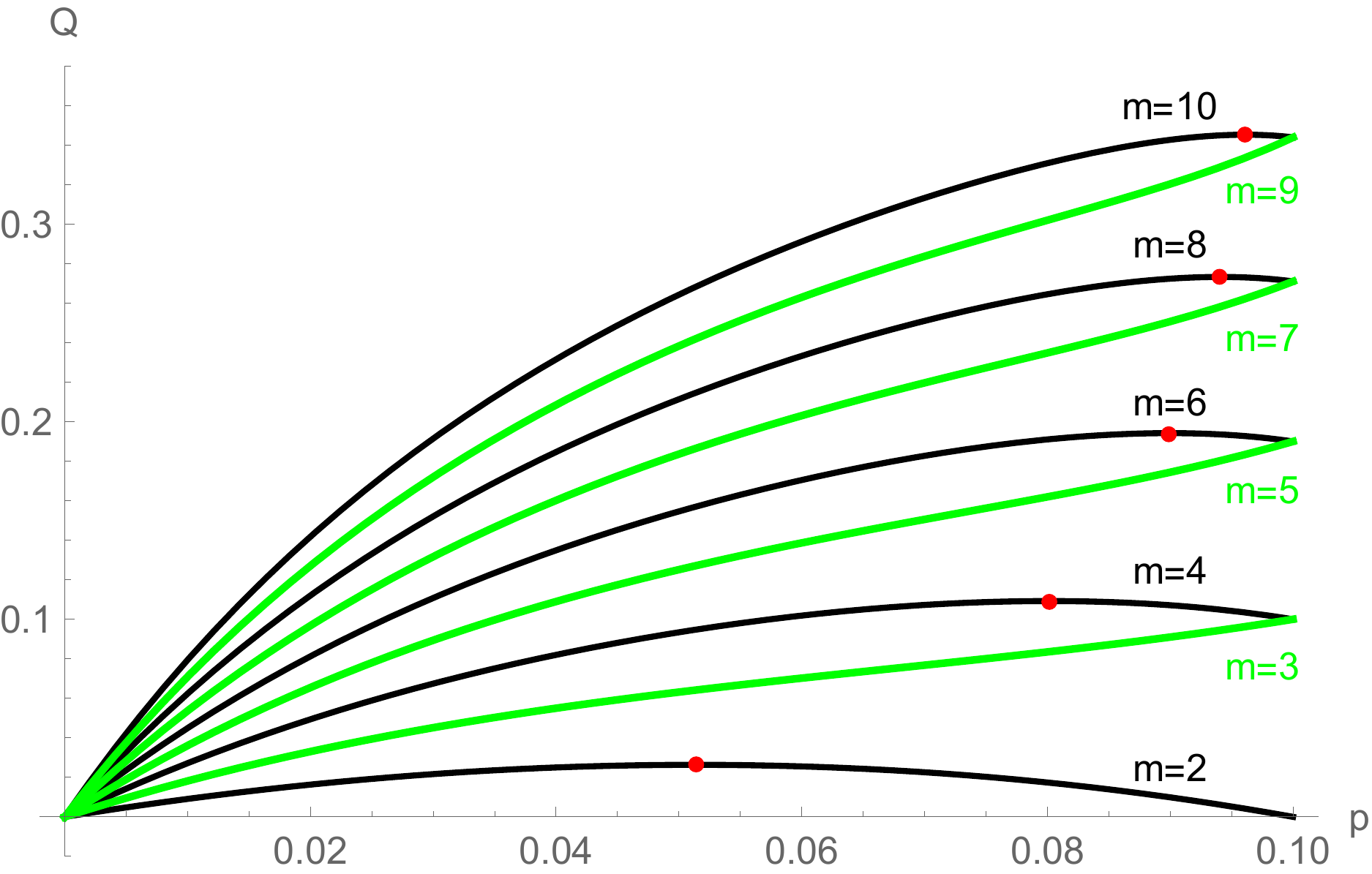}
\end{center}
\caption{Plots of $Q$, for $n=10$, $m=2,\dots,10$, $p\leq \frac{1}{n}$. $\hat p$ are shown by red dots.}
\label{Figure2}
\end{figure}

Note that for even $m$ we get concave plots with interior maxima, while for odd $m$ we get increasing plots maximised at $p=\frac{1}{n}$, that is the maximum value that $p$ can get. For the graphs with even $m$, we also plot the maxima of $p$ in red dots. Note that these maxima move to the right as $m$ increases. For even $m\geq 6$ there is no closed solution for the optimal value of $p$, however (as shown in Figure \ref{Figure2}) we can numerically optimise the expression \eqref{eq15*} in terms of $p$. Note as well that for $p=\frac{1}{n}$ (and $s=1$) the interception probability $Q$ is the same for $m=3$ and $m=4$ (or $m=2k-1$, $m=2k$), which illustrates equality in the right inequality of \eqref{eq11}.

In Figure \ref{Figure2} we numerically calculated the optimal values $\hat{p}$ for a single value of $n=10$. We do the same in Figure \ref{Figure3} (left) but for several values of $n$, namely for $n=5$ (blue circles), $n=10$ (orange squares), $n=15$ (green diamonds), $n=20$ (red triangles). For each $n$ we draw in horizontal dashed lines the asymptotes $1/n$. The decrease in $\hat{p}$ for fixed $m$ as $n$ increases is mainly due to the fact that when leaving the center the Patroller divides her probability of leaving in $n$ possible locations. To reduce this type of dependence on $n$, we plot in Figure \ref{Figure3} (right) the optimal probability $\hat{r}=1-n\hat{p}$ which is seen to depend only slightly on $n$, as the four corresponding symbols are tightly packed for fixed $m$. Note that for $m=4$ the values of $\hat{r}$ are clustered near the value of $\hat{r}\approx 0.202$ given by the asymptotic analysis of Section \ref{Section2.4} (see \eqref{11*}).

\begin{figure}[H]
	\centering
\includegraphics[width=0.475\linewidth]{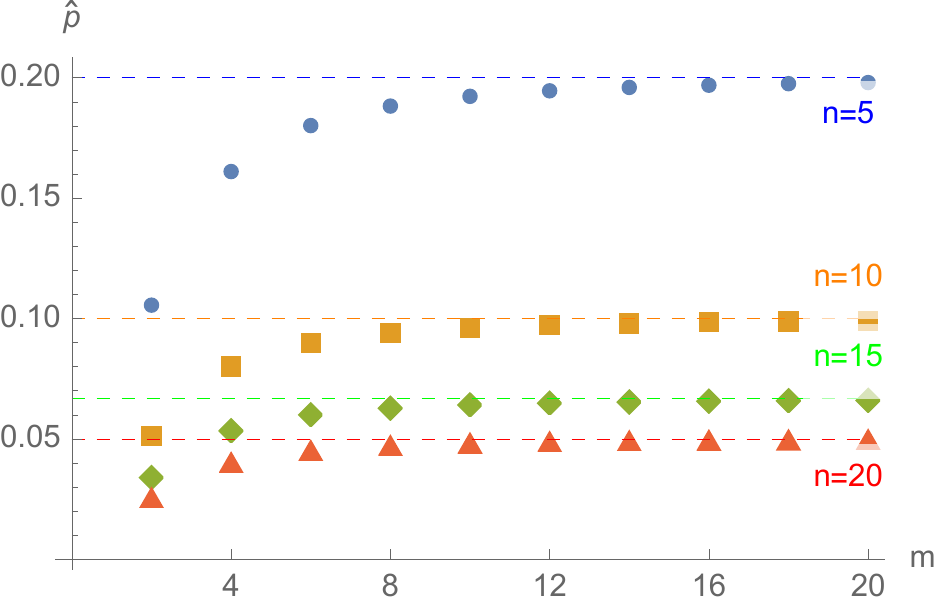}\quad
\includegraphics[width=0.475\linewidth]{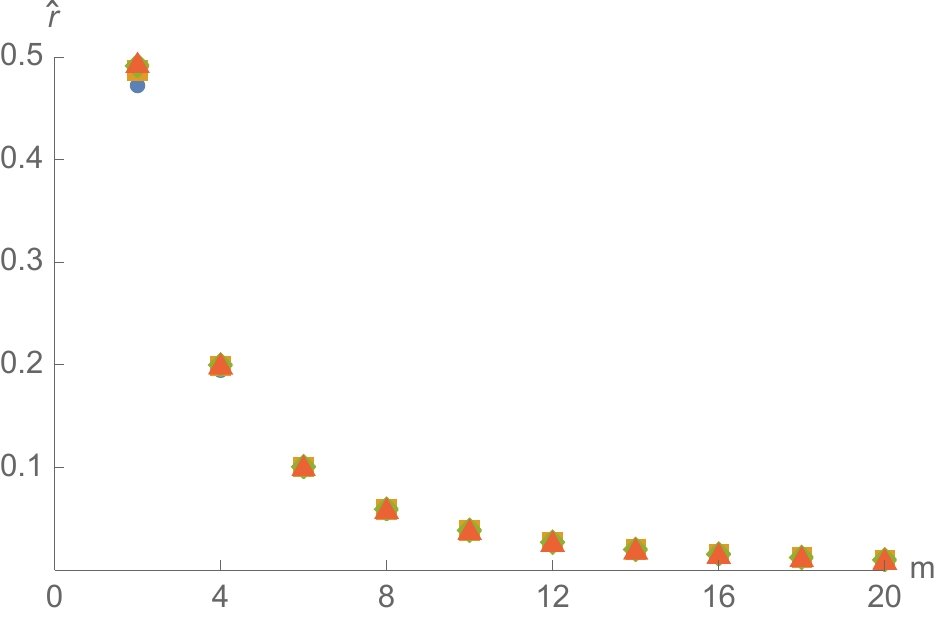}
\caption{$\hat{p}$ (left) and $\hat{r}$ (right) for even $m$, $n=5,\,10,\,15,\,20$\,.}
\label{Figure3}
\end{figure}

\subsection{The Cost of Wearing a Uniform}\label{Section2.7}
\indent 

In the above analysis, the Attacker uses his information about the presence and absence of the Patroller at his location in optimizing the timing of his attack. So it would seem intuitive that forcing the Patroller
to wear a uniform reduces the optimal probability of intercepting the attack. However in a distant but related multiple Patroller game of \cite{Lin-3} it is shown that wearing a uniform does not affect this probability (value of the game). So motivated by that paper we decided to compare the two values (with and without a uniform) in our model. We find in all cases that there is indeed a cost of wearing a uniform, but that this cost varies with the parameters of our game. In Lin's model, a Patroller who leaves the Attacker's position will never return, so the Attacker is only vulnerable to the next Patroller, and this may in part be why there is no cost of wearing a uniform in that model. 

Consider the original patrolling game where the Patroller is non-uniformed and Markovian. Assuming the attack duration is $m=2$, in any interval of length equal to $2$ she visits only one end, so the value of the game (interception probability) is given by $\tilde{V}=\frac{1}{n}$. Recall that the corresponding interception probability $V$ for the Uniformed Patroller game is given by \eqref{eq2}. For large $n$, taking the ratio $\frac{V}{\tilde{V}}$ (uniformed over non-uniformed) and applying twice L'Hopital's rule when indeterminate forms occur, we get
\begin{align*}
\lim_{n\to \infty}\frac{V}{\tilde{V}}&=\lim_{n\to \infty}\frac{(2\cdot n -1)-2\cdot\sqrt{n\cdot (n-1)}}{\frac{1}{n}}=\lim_{n\to \infty}\frac{n\cdot \bigl(2-\frac{1}{n}-2\cdot \sqrt{1-\frac{1}{n}}\bigr)}{n\cdot \frac{1}{n^2}}=\\&=\lim_{n\to \infty}\frac{\frac{1}{n^{2}}-\bigl(1-\frac{1}{n}\bigr)^{-\frac{1}{2}}\cdot \frac{1}{n^{2}}}{-2\cdot \frac{1}{n^{3}}}=\lim_{n\to\infty}\frac{\frac{1}{n^{2}}\cdot \Big(1-\bigl(1-\frac{1}{n}\bigr)^{-\frac{1}{2}}\Bigr)}{\frac{1}{n^{2}}\cdot \bigl(-2\cdot \frac{1}{n}\bigr)}=\\&=\lim_{n\to\infty}\frac{1/2\cdot \bigl(1-\frac{1}{n}\bigr)^{-3/2}\cdot \frac{1}{n^2}}{2\cdot \frac{1}{n^{2}}}=\frac{1}{4}\,.
\end{align*}

\begin{proposition}
For large $n$ and attack duration $m=2$, wearing a uniform reduces the interception probability (value) by a factor of $4$.
\end{proposition}
 
Similar results can be obtained if we consider alternative values of $m$. For example, recall that when $m$ is odd the value $V$ of the Uniformed Patroller game on the star $S_{n}$ is given by \eqref{eq6}. Without a uniform the Patroller starts equiprobably a random walk at $t=1$ either at an end or at the center, while the Attacker also equiprobably initiates his attack either at an odd or at an even period. That is, the interception probability in the case when $m=2\cdot j+1$ is given by 

\begin{gather*}
\tilde{V}=\frac{1}{2}\cdot\biggl(1-\Bigl(\frac{n-1}{n}\Bigr)^{j}\biggr)+\frac{1}{2}\cdot\biggl(1-\Bigl(\frac{n-1}{n}\Bigr)^{j+1}\biggr).
\end{gather*}

We look again at the ratio of the two interception probabilities. It is easy to show that as $n$ goes to infinity $\frac{V}{\tilde{V}}$ converges to $\frac{2\cdot j}{(2\cdot j+1)}=\frac{(m-1)}{m}$ (the proof is similar to the one for $m=2$). 

\begin{proposition}
When the attack duration m is odd, the relative loss in interception probability of wearing a uniform for large n is the reciprocal of the attack duration. That is,
\[\lim_{n\to \infty} \frac{\tilde{V}-V}{\tilde{V}}= \frac{1}{m}.\]
\end{proposition}

\section{Conclusion}\label{Section3}
\indent 

This paper examined the aspects of patrolling games that change when the Attacker can detect the presence of the Patroller at his location. We call such problems Uniformed Patroller Games, as the wearing of a uniform is a metaphor for this type of detection. The Patroller’s problem does not change much from the earlier (non-uniformed) version, but the Attacker must now decide how long to wait after the Patroller has left the location where he waits to carry out his attack. We determine solutions to this game when the Patroller must defend an arbitrary number of locations and the attack lasts an arbitrary number of periods. We proved that for any values of the parameters $m$ and $n$, an  optimal delay is $d=2$ and it is optimal to never patrol a location in consecutive periods ($s=1$). We note that it has been shown for some graphs that other values of the delay $d$ are optimal, see \cite{Alp-3}.

In this model the game without a uniform is easy to solve so that we can determine the reduction in the probability of intercepting an attack when wearing a uniform. This reduction turns out to be sensitive to the parameters involved. These results contrast with the lack of any reduction in interception probability in the continuous time model of \cite{Lin-3}, with multiple agents (of course the models are very different). Since wearing a uniform makes intercepting attacks more difficult, we plan to investigate the role the uniform plays in deterring attacks. This should add to the current lively debate on the importance of uniformed police ‘on the beat’ as opposed to undercover agents.

%
%
%


\section*{Acknowledgements}

\end{document}